\documentclass[11pt]{article}
\usepackage{epsfig}
\usepackage{amsmath}

\newtheorem{theorem}{Theorem}
\newtheorem{lemma}{Lemma}
\newtheorem{remark}{Remark}
\newtheorem{proposition}{Proposition}

\newenvironment{proof}
{\begin{trivlist} \item[]{\bf Proof. }}%
{\hspace*{\fill}$\rule{.3\baselineskip}{.35\baselineskip}$\end{trivlist}}

\usepackage{amssymb}

\begin{document}

\def\myint{\int \!\! d^{ {\scriptscriptstyle D} } \boldsymbol{x} \; }


\title{Bifurcation of gap solitons in periodic potentials \\ with
a sign-varying nonlinearity coefficient}

\author{Juan Belmonte-Beitia$^{a,b}$ and Dmitry Pelinovsky$^{a}$ \\
{\small $^a$ Department of Mathematics, McMaster University, Hamilton ON, Canada, L8S 4K1} \\
{\small $^b$ Departamento de Matem\'aticas, E. T. S. de Ingenieros
Industriales} \\ {\small and Instituto de Matem\'atica Aplicada a
la Ciencia y la Ingenier\'{\i}a (IMACI)} \\{\small Universidad de
Castilla-La Mancha, Ciudad Real, Spain,  13071} }

\date{\today}
\maketitle

\begin{abstract}
We address the Gross--Pitaevskii (GP) equation with a periodic linear potential and
a periodic sign-varying nonlinearity coefficient. Contrary to the claims in the
previous works of Abdullaev {\em et al.} [PRE {\bf 77}, 016604 (2008)]
and Smerzi \& Trombettoni [PRA {\bf 68}, 023613 (2003)], we show that
the intersite cubic nonlinear terms in the discrete nonlinear Schr\"{o}dinger (DNLS) equation
appear beyond the applicability of assumptions of the tight-binding approximation.
Instead of these terms, for an even linear potential and an odd nonlinearity coefficient, 
the DNLS equation and other reduced equations for the semi-infinite gap 
have the quintic nonlinear term, which correctly describes bifurcation of gap solitons.
\end{abstract}



\section{Introduction}

The generalized DNLS (discrete nonlinear Schr\"{o}dinger) equation with intersite cubic nonlinear terms,
\begin{eqnarray}
\nonumber i \dot{c}_n & = & \alpha (c_{n+1} + c_{n-1}) + \beta |c_n|^2 c_n \\
\nonumber & \phantom{t} & + \gamma (2 |c_n|^2 (c_{n+1} + c_{n-1})
+ c_n^2 (\bar{c}_{n+1} + \bar{c}_{n-1}) + |c_{n+1}|^2 c_{n+1} +
|c_{n-1}|^2 c_{n-1}) \\ \label{DNLS}  & \phantom{t} & + \delta
((c_{n+1}^2 + c_{n-1}^2) \bar{c}_n + 2 (|c_{n+1}|^2 + |c_{n-1}|^2)
c_n),
\end{eqnarray}
where $(\alpha,\beta,\gamma,\delta)$ are constant parameters and
the dot denotes differentiation in time, was
derived independently in various contents. Smerzi \& Trombettoni
\cite{ST} suggested that this equation models Bose--Einstein
condensates in a lattice, when Wannier functions
associated with a periodic potential are replaced by the nonlinear
bound states. Independently, this equation was derived
heuristically by Oster {\em et al.} \cite{OJE03} to model
waveguide arrays in a nonlinear photonic crystal. Earlier, the
same equation was obtained by Claude {\em et al.} \cite{CKKS93}
for modelling of slowly varying discrete breathers in the
Fermi--Pasta--Ulam lattices using asymptotic multi-scale
expansions. Very recently, the generalized DNLS equation was rederived again
by Abdullaev {\em et al.} \cite{ABDKK} in a more specific context
of the GP (Gross--Pitaevskii) equation with a periodic potential and a
periodic sign-varying nonlinearity coefficient. If the coefficient
in front of the onsite cubic nonlinear term of the DNLS equation
vanishes thanks to the sign-varying nonlinearity (that is, if $\beta = 0$ 
in (\ref{DNLS})), the authors of
\cite{ABDKK} incorporated other intersite cubic nonlinear terms
from a decomposition involving Wannier functions.

In what follows, we focus on the specific applications of the
generalized DNLS equation (\ref{DNLS}) in the context of
Bose--Einstein condensates in a lattice. Therefore, we consider
the GP equation with a periodic linear potential and a
periodic sign-varying nonlinearity coefficient in the form
\begin{equation}\label{INLSEdt}
i \partial_t \Psi = - \partial_x^2 \Psi + V(x)\Psi + G(x) |\Psi|^{2}\Psi,
\end{equation}
where $V(x)$ and $G(x)$ are smooth, $2\pi$-periodic functions on $\mathbb{R}$.
To make all arguments precise, we assume that
\begin{equation}
\label{potentials}
V(-x)=V(x), \quad G(-x) = -G(x), \quad x \in \mathbb{R}.
\end{equation}
In this case, $\beta = 0$ and our main result states that 
the intersite cubic nonlinear terms in the
generalized DNLS equation (\ref{DNLS}) appear beyond the
applicability of the DNLS equation in the tight-binding
approximation and hence must be dropped from the leading order of
the asymptotic equation. Instead of these terms, the onsite
quintic nonlinear term must be taken into account
to balance the linear dispersion term in the
quintic DNLS equation
\begin{eqnarray}
\label{DNLSmodified} i \dot{c}_n = \alpha (c_{n+1} + c_{n-1})
+ \chi |c_n|^4 c_n,
\end{eqnarray}
where $(\alpha,\chi)$ are constant parameters which can be
computed from analysis of the GP equation
(\ref{INLSEdt}) with potentials (\ref{potentials}).

Note that the approach leading to the DNLS equation is general and can be 
applied to other $2\pi$-periodic functions $V(x)$ and $G(x)$. In a general 
case, $\beta \neq 0$ and the onsite cubic nonlinear term is the only nonlinear term, which 
must be accounted in the cubic DNLS equation 
\begin{eqnarray}
i \dot{c}_n = \alpha (c_{n+1} + c_{n-1}) + \beta |c_n|^2 c_n, \label{DNLScubic}
\end{eqnarray}
at the leading order of the asymptotic expansions.

To compare the outcomes of the generalized DNLS equation
(\ref{DNLS}) with those of the quintic DNLS equation
(\ref{DNLSmodified}), we study bifurcations of gap solitons in the
semi-infinite band gap. We show analytically that $\alpha$ and
$\chi$ has equal {\em negative} signs in the semi-infinite band
gap so that the quintic DNLS equation (\ref{DNLSmodified}) always
has a ground state, indicating that bifurcation of a
gap soliton in the semi-infinite gap always occurs in
the GP equation (\ref{INLSEdt}) with potentials (\ref{potentials}). 
Recall that this bifurcation does not
occur if the nonlinearity coefficient is sign-definite and
positive, see Pankov \cite{Pankov}. 

In contrary to the predictions
of the quintic DNLS equation (\ref{DNLSmodified}), we also show
that the corresponding version of the generalized DNLS equation
(\ref{DNLS}) does not admit localized solutions for any values of 
$\alpha$ and $\gamma$ (when $\beta = \delta = 0$) at least in the
slowly varying approximation. A numerical test with particular
potentials
\begin{equation}
\label{explicit-potential} V(x) = V_0 (1 - \cos(x)), \quad G(x) =
G_0 \sin(x),
\end{equation}
indicates that the gap solitons {\em do} exist in the
semi-infinite gap independently of the signs of $V_0$ and $G_0$.
The rigorous proof of existence of localized solutions in the GP
equation (\ref{INLSEdt}) with potentials (\ref{potentials})
in the semi-infinite band gap is beyond the scopes of this work and is
a subject of an ongoing work \cite{Pankov-personal}.

We also inspect another asymptotic reduction of the
GP equation (\ref{INLSEdt}) to the continuous
nonlinear Schr\"{o}dinger (CNLS) equation, see review of
asymptotic reductions of the GP equation with a periodic
potential in Pelinovsky \cite{Pel}. We show that the
corresponding CNLS equation also has a focusing quintic nonlinear
term, which supports the same conclusion on bifurcation of a
gap soliton in the semi-infinite gap.

We note that reductions to the DNLS and CNLS equations were
recently justified with rigorous analysis both in the stationary
and time-dependent cases, see works \cite{Sch1,Sch2} in the
context of the DNLS equation and works \cite{Sch3,Sch4} in the context
of the CNLS equation. Therefore, it is a matter of a routine technique
to formalize arguments of our paper.

We shall add that the literature on the GP equation (\ref{INLSEdt}) is rapidly growing 
in physics literature. The GP equation with a periodic
nonlinearity coefficient was considered by Fibich {\em at al.}
\cite{Fibich}, where no linear potential $V(x)$ was included and the mean value
of $G(x)$ was assumed to be nonzero. For the same equation, Sakaguchi \& Malomed 
\cite{SakMal} derived a quintic CNLS equation in a slowly varying approximation 
of a broad soliton. 

A more general equation with both linear and nonlinear 
periodic coefficients was studied by Bludov {\em et al.} in \cite{Bludov1,Bludov2}, 
where gap solitons were approximated numerically. It was shown in these works that bifurcations 
of small-amplitude gap solitons near the lowest band edge depend on the sign of 
the cubic coefficient in the effective CNLS equation. Using perturbation theory, 
Rapti {\em et al.} \cite{Rapti}
studied existence and stability of gap solitons in the semi-infinite gap for 
the GP equation with small linear and nonlinear periodic coefficients.

The paper is organized as follows. Section 2 justifies the asymptotic reduction
of the Gross--Pitaevskii equation (\ref{INLSEdt}) to the quintic DNLS equation
(\ref{DNLSmodified}). Section 3 gives results on existence of stationary
localized modes in the generalized DNLS equations (\ref{DNLS})
and (\ref{DNLSmodified}) and discusses the relevance of previous works
\cite{ST} and \cite{ABDKK}. Section 4 justifies the asymptotic reduction
to the quintic CNLS equation.

{\bf Acknowledgements:} J. B.-B. has been partially supported by grants  PCI08-0093 
(Consejer\'{\i}a de Educaci\'on y Ciencia de la Junta de Comunidades de Castilla-La Mancha, Spain), 
PRINCET and FIS2006-04190 (Ministerio de Educaci\'on y Ciencia, Spain). J. B.-B. also would like to thank the Mathematics Department at McMaster University for their hospitality during his visit there.

\section{Reductions to the DNLS equation}

To consider the tight-binding approximation and reductions to the
DNLS equation, we assume that
\begin{equation}
\label{potential-eps}
V(x) = \epsilon^{-2} V_0(x),
\end{equation}
where $\epsilon$ is a small parameter and $V_0$ is a
smooth, $2\pi$-periodic, and even function on $\mathbb{R}$. In
what follows and without loss of generality, we set
\begin{equation}
\label{potential-normalization}
V_0(0) = 0  \quad \mbox{\rm and} \quad V''_0(0) = 2,
\end{equation}
so that $V_0(x) = x^2 + {\cal O}(x^4)$ as $x \to 0$. For
particular explicit computations, we consider the standard example
$$
V_0(x) = 2(1 - \cos(x)) = 4 \sin^2 \left(\frac{x}{2}\right). 
$$
The limit $\epsilon \to 0$ is generally
referred to as the semi-classical limit \cite{Helffer}.

Let $\Psi(x;k)$ be the Bloch function of
\begin{equation}
\label{Bloch} L \Psi(x;k) = E(k) \Psi(x;k), \quad L =
-\partial_x^2 + \epsilon^{-2} V_0(x),
\end{equation}
for the lowest energy band function $E(k)$. It is known
(see review in Pelinovsky {\em et al.} \cite{Sch1})
that $E(k)$ and
$\Psi(x,k)$ satisfy
$$
E(k) = E(k+1) = E(-k), \quad k \in \mathbb{R}
$$
and
$$
\Psi(x;k) = e^{-2 \pi k i} \Psi(x+2\pi;k) = \Psi(x;k+1) = \bar{\Psi}(x;-k), \quad x
\in \mathbb{R}, \quad k \in \mathbb{R},
$$
so that one can define the Fourier series decompositions
$$
E(k) = \sum_{n \in \mathbb{Z}} \hat{E}_n e^{2\pi n k i}, \quad
\Psi(x;k) = \sum_{n \in \mathbb{Z}} \hat{\psi}_n(x) e^{2\pi n k i},
$$
with real-valued Fourier coefficients satisfying the reduction
$$
\hat{E}_{n} = \hat{E}_{-n}, \quad \hat{\psi}_n(x) = \hat{\psi}_0(x -
2\pi n), \quad n \in \mathbb{Z}.
$$
Functions $\{ \hat{\psi}_n(x) \}_{n \in \mathbb{Z}}$ are referred to as the Wannier functions.
For the lowest energy band, these functions
form an orthonormal basis in a subspace of $L^2(\mathbb{R})$ associated with 
the lowest energy band, enjoy an
exponential decay to zero as $|x| \to \infty$ and satisfy the
system of differential equations
\begin{equation}
\label{Wannier} \left( L - \hat{E}_0 \right) \hat{\psi}_0(x) = \sum_{n
\geq 1} \hat{E}_{n} \left( \hat{\psi}_n(x) + \hat{\psi}_{-n}(x) \right), \quad x \in \mathbb{R}.
\end{equation}
Thanks to orthogonality and normalization of the
Wannier functions, we infer that $\hat{E}_n$ can be computed from the overlapping integrals
\begin{equation}
\label{overlaping} \hat{E}_n = \langle L \hat{\psi}_0, \hat{\psi}_n
\rangle = \int_{\mathbb{R}} \left[ \hat{\psi}_0'(x) \hat{\psi}_n'(x) +
\epsilon^{-2} V_0(x) \hat{\psi}_0(x) \hat{\psi}_n(x) \right] dx,
\quad n \in \mathbb{N}.
\end{equation}

For the semi-infinite gap and for even potentials, Wannier
functions $\{ \hat{\psi}_n(x) \}_{n \in \mathbb{Z}}$ are strictly
positive and even on $\mathbb{R}$. It is proved with the standard
technique in the semi-classical limit $\epsilon \to 0$ (see review in
Aftalion \& Helffer \cite{AH09})
that the Wannier function $\hat{\psi}_0(x)$ can be approximated
near $x = 0$ by the normalized Gaussian eigenfunction of
$$
\left( -\partial_x^2 + \frac{x^2}{\epsilon^2} \right) \psi_0(x) =
\frac{1}{\epsilon} \psi_0(x), \quad x \in \mathbb{R},
$$
or explicitly,
\begin{equation}
\label{Gaussian}
\psi_0(x) = \frac{1}{(\pi \epsilon)^{1/4}}
e^{-\frac{x^2}{2\epsilon}}, \quad x \in \mathbb{R}.
\end{equation}
This approximation suggests that
\begin{equation}
\label{leading-order} \hat{E}_0 \sim \frac{1}{\epsilon}, \quad
\hat{\psi}_0(x) \sim \psi_0(x), \quad \mbox{\rm near} \quad x = 0,
\end{equation}
where we have used the notation $A(\epsilon) \sim B(\epsilon)$ for
two functions of $\epsilon$ near $\epsilon = 0$ 
to indicate that $A(\epsilon)/B(\epsilon) \to 1$ as $\epsilon \to 0$.
To obtain approximations for the overlapping integrals
(\ref{overlaping}), one need to proceed with the WKB solution
\begin{equation}
\label{WKBsolution} \hat{\psi}_0(x) \sim A(x) e^{-\frac{1}{\epsilon}
\int_0^x S(x') dx'}, \quad x \in (0,2\pi),
\end{equation}
where
\begin{eqnarray*}
S(x) & = & \sqrt{V_0(x)}, \\
A(x) & = & \frac{1}{(\pi \epsilon)^{1/4}} \exp\left[ \int_0^x
\frac{1 - S'(x')}{2 S(x')} dx' \right], \quad x \in (0,2\pi).
\end{eqnarray*}
The WKB solution (\ref{WKBsolution}) is derived by neglecting the term $A''(x)$ in the
left-hand-side of (\ref{Wannier}) and by dropping the right-hand-side of
(\ref{Wannier}) thanks to the hierarchy of overlapping integrals in
\begin{equation}
\label{hierarchy}
\ldots \ll |\hat{E}_2| \ll |\hat{E}_1 | \ll |\hat{E}_0|.
\end{equation}
In addition, to derive the explicit expression for $A(x)$ we
have replaced $\hat{E}_0$ by $1/\epsilon$ and used the matching
condition of $\hat{\psi}_0(x)$ with $\psi_0(x)$ as $x \downarrow 0$. Note that
the expression for $A(x)$ diverges as $x \uparrow 2\pi$.

Thanks to the explicit formulas and the symmetry of $\hat{\psi}_0(x)$
on $\mathbb{R}$, the first overlapping integral is computed as
follows
\begin{eqnarray*}
\hat{E}_1 & = & 2 \int_{-\infty}^{\pi} \hat{\psi}_0(x) \left(
-\partial_x^2 + \epsilon^{-2} V_0(x) - \hat{E}_0 \right)
\hat{\psi}_0(x-2\pi) dx \\
& = & 4 \hat{\psi}_0(\pi) \hat{\psi}_0'(\pi) + 2 \int_{-\infty}^{\pi}
\hat{\psi}_0(x-2\pi) \left( -\partial_x^2 + \epsilon^{-2} V_0(x) -
\hat{E}_0 \right) \hat{\psi}_0(x) dx.
\end{eqnarray*}
Neglecting the integral (thanks again to smallness of the
right-hand-side of (\ref{Wannier}) on $(-\infty,-\pi]$) and
substituting the WKB solution (\ref{WKBsolution}) at $x = \pi$, we
infer that the leading order of the first overlapping integral is given by
\begin{equation}
\label{hat-E-1}
\hat{E}_1 \sim 4 \hat{\psi}_0(\pi) \hat{\psi}_0'(\pi) = 
-\frac{4 \sqrt{V_0(\pi)}}{\pi^{1/2} \epsilon^{3/2}}
\exp\left(-\frac{2}{\epsilon} \int_0^{\pi} \sqrt{V_0(x)} dx +
\int_0^{\pi}\frac{1 - S'(x)}{S(x)} dx \right).
\end{equation}
For instance, if $V_0(x) = 2(1 - \cos(x))$, then
$$
S(x) = 2 \sin\left(\frac{x}{2}\right), \quad
A(x) = \frac{1}{(\pi \epsilon)^{1/4} \cos\left(\frac{x}{4}\right)}, \quad
x \in (0,2\pi),
$$
so that
$$
\hat{E}_1 \sim -\frac{16}{\pi^{1/2} \epsilon^{3/2}}
e^{-\frac{8}{\epsilon}}.
$$
Similarly, one can establish the hierarchy of other overlapping integrals in (\ref{hierarchy}).
See Helffer \cite{Helffer} for rigorous justification of the WKB solutions above.

To deal with nonlinear terms, we compute the integral involving
$G(x) \hat{\psi}^4_0(x)$ as $\epsilon \to 0$. This integral
can be computed with the use of the Gaussian approximation
(\ref{Gaussian})--(\ref{leading-order}), thanks to the fast decay
of $\hat{\psi}_0(x)$ on $\mathbb{R}$ and the smoothness of $G(x)$ on $\mathbb{R}$:
\begin{equation}
\label{psi-0-4}
\int_{\mathbb{R}} G(x) \hat{\psi}_0^4(x) dx \sim \frac{1}{\pi \epsilon}
\int_{\mathbb{R}} G(x) e^{-\frac{2x^2}{\epsilon}} dx \sim \frac{1}{(2 \pi \epsilon)^{1/2}} G(0).
\end{equation}
The overlapping integrals involving homogeneous quartic powers
of $\hat{\psi}_0(x)$, $\hat{\psi}_0(x-2\pi)$, etc. are
much smaller compared to the integral (\ref{psi-0-4}), thanks again to the fast decay 
of $\hat{\psi}_0(x)$ on $\mathbb{R}$.

\subsection{Reduction to the cubic DNLS equation}

Asymptotic reduction to the cubic DNLS equation holds for $G(0) \neq 0$.
Computations (\ref{hat-E-1}) and (\ref{psi-0-4}) suggest the use of the scaling transformation
$$
\Psi(x,t) = \epsilon^{1/4} \mu^{1/2} \left( \Psi_0 + \mu \Psi_1 \right) e^{-i
\hat{E}_0 t},
$$
with a new small parameter
\begin{equation}
\label{mu}
\mu = \frac{1}{\pi^{1/2} \epsilon^{3/2}} e^{-\frac{2}{\epsilon}
\int_0^{\pi} \sqrt{V_0(x)} dx},
\end{equation}
for asymptotic solutions of the Gross--Pitaevskii equation (\ref{INLSEdt}).
To give main details, let $T = \mu t$ be slow time
and decompose
$$
\Psi_0 = \sum_{n \in \mathbb{Z}} c_n(T) \hat{\psi}_n(x),
$$
for some coefficients $\{ c_n \}_{n \in \mathbb{Z}}$ to be
defined. The remainder term $\Psi_1$ satisfies
\begin{eqnarray*}
i \partial_t \Psi_1 & = & (L - \hat{E}_0) \Psi_1 + \sum_{n \in
\mathbb{Z}} \left( -i \dot{c}_n + \mu^{-1} \sum_{m \in \mathbb{N}}
\hat{E}_m \left( c_{n+m} + c_{n-m}\right) \right) \hat{\psi}_n \\
& \phantom{t} & + \epsilon^{1/2} G(x) |\Psi_0 + \mu \Psi_1|^2 (\Psi_0 + \mu \Psi_1).
\end{eqnarray*}
Coefficients $\{ c_n \}_{n \in \mathbb{Z}}$ are uniquely defined by
the orthogonality condition
$$
\langle \hat{\psi}_n, \Psi_1 \rangle = 0  \quad \mbox{\rm for all}
\;\; n \in \mathbb{Z},
$$
which ensures that $\Psi_1$ is in the orthogonal complement
of the subspace of $L^2(\mathbb{R})$ corresponding to the lowest spectral band
of operator $L$. Orthogonal projections to $\{ \hat{\psi}_n \}_{n \in \mathbb{Z}}$
truncated at the leading-order terms as
$\mu \to 0$ take the form of the cubic DNLS equation
\begin{eqnarray}
\label{DNLScub} i \dot{c}_n = \alpha (c_{n+1} + c_{n-1}) + \beta |c_n|^2 c_n,
\end{eqnarray}
where
\begin{eqnarray}
\label{alpha}
\alpha & = & \mu^{-1} \hat{E}_1 \sim -4 \sqrt{V_0(\pi)}
\exp\left(\int_0^{\pi}\frac{1 - S'(x)}{S(x)} dx \right), \\
\label{beta}
\beta & = & \epsilon^{1/2} \int_{\mathbb{R}} G(x) \hat{\psi}_0^4(x) dx
\sim \frac{1}{(2 \pi)^{1/2}} G(0),
\end{eqnarray}
thanks to the fact that other overlapping integrals in the linear and cubic terms are smaller. 
Rigorous justification of the cubic DNLS equation (\ref{DNLScub}) on a finite time interval 
is proved by Pelinovsky \& Schneider \cite{Sch2}, where the main result is formulated in space 
${\cal H}^1(\mathbb{R})$ defined by the norm $\| \Psi \|_{{\cal H}^1(\mathbb{R})} := 
\sqrt{\langle (L + I) \Psi, \Psi \rangle}$, where $L = -\partial_x^2 + V(x)$. 

\begin{theorem}
Assume that $V(x)$ is given by (\ref{potential-eps}), $G(0) \neq 0$, and $\mu$ is given by (\ref{mu}).  
Let $\{ c_n(T) \}_{n \in \mathbb{Z}} \in C^1(\mathbb{R},l^1(\mathbb{Z}))$
be a global solution of the cubic DNLS equation (\ref{DNLScub}) 
with initial data $\{ c_n(0) \}_{n \in \mathbb{Z}} \in l^2_p(\mathbb{Z})$
for any $p > \frac{1}{2}$.
Let $\Psi_0 \in {\cal H}^1(\mathbb{R})$ satisfy the bound
$$
\left\| \Psi_0 - \epsilon^{1/4} \mu^{1/2}
\sum_{n \in \mathbb{Z}} c_n(0) \hat{\psi}_n \right\|_{{\cal
H}^1(\mathbb{R})} \leq C_0 \epsilon^{1/4} \mu^{3/2},
$$
for some $C_0 > 0$. There exists $\mu_0 > 0$, $T_0 > 0$, and $C > 0$,
such that for any $\mu \in (0,\mu_0)$, the GP equation (\ref{INLSEdt}) 
with initial data $\Psi(0) = \Psi_0$ has a solution 
$\Psi(t) \in C([0,T_0/\mu],{\cal H}^1(\mathbb{R}))$ satisfying
the bound
$$
\forall t \in \left[ 0, T_0/\mu \right] : \quad
\left\| \Psi(\cdot,t) - \epsilon^{1/4} \mu^{1/2} \left( \sum_{n \in \mathbb{Z}}
c_n(\mu t) \hat{\psi}_n \right) e^{-i \hat{E}_0 t} \right\|_{{\cal H}^1(\mathbb{R})} \leq C
\epsilon^{1/4} \mu^{3/2}.
$$
\label{theorem-1}
\end{theorem}

\subsection{Reduction to the quintic DNLS equation}

If $G(0) = 0$, and we assume that $G(x)$ is odd on $\mathbb{R}$, then 
$\beta = 0$ and the DNLS equation
(\ref{DNLScub}) becomes a linear equation. A modified asymptotic solution is needed to incorporate
the leading order of the asymptotic expansion. We will show that the modified scaling
$$
\Psi(x,t) = \epsilon^{-1/4} \mu^{1/4} \left( \Psi_0 + \mu^{1/2} \Psi_1 + \mu \Psi_2 \right) e^{-i
\hat{E}_0 t},
$$
will reduce the GP equation (\ref{INLSEdt}) to the quintic DNLS
equation (\ref{DNLSmodified}) if $G'(0) \neq 0$. Again, let $T =
\mu t$ be the slow time and define $\Psi_0$ and $\Psi_1$ by
$$
\Psi_0 = \sum_{n \in \mathbb{Z}} c_n(T) \hat{\psi}_n(x), \quad
\Psi_1 = \sum_{n \in \mathbb{Z}} |c_n(T)|^2 c_n(T) \hat{\varphi}_n(x),
$$
where $\hat{\varphi}_n(x) = \hat{\varphi}_0(x-2\pi n)$, $n \in \mathbb{Z}$ is a solution of
\begin{equation}
\label{inhomogeneous-problem}
( L - \hat{E}_0 ) \hat{\varphi}_0(x) = - \epsilon^{-1/2} G(x) \hat{\psi}_0^3(x), \quad x \in \mathbb{R},
\end{equation}
under the orthogonality condition
\begin{equation}
\label{nonlinearity-0}
\int_{\mathbb{R}} G(x) \hat{\psi}_0^4(x) dx = 0.
\end{equation}
The remainder term $\Psi_2$ satisfies
\begin{eqnarray*}
i \partial_t \Psi_2 & = & (L - \hat{E}_0) \Psi_2 + \sum_{n \in
\mathbb{Z}} \left( -i \dot{c}_n + \mu^{-1} \sum_{m \in \mathbb{N}}
\hat{E}_m \left( c_{n+m} + c_{n-m}\right) \right) \hat{\psi}_n \\
& \phantom{t} & -
i \mu^{1/2} \sum_{n \in \mathbb{Z}} \frac{d}{dT} \left( |c_n|^2 c_n
\right) \hat{\varphi}_n + \epsilon^{-1/2} \mu^{-1/2} G(x) \\
& \phantom{t} & \times \left(
|\Psi_0 + \mu^{1/2} \Psi_1 + \mu \Psi_2|^2 (\Psi_0 + \mu^{1/2} \Psi_1 + \mu \Psi_2)
-\sum_{n \in \mathbb{Z}} |c_n|^2 c_n \hat{\psi}_n^3 \right).
\end{eqnarray*}
Orthogonal projections to $\{ \hat{\psi}_n \}_{n \in \mathbb{Z}}$ truncated
at the leading-order terms as $\mu \to 0$ result in the quintic DNLS equation (\ref{DNLSmodified}) with
the same expression for $\alpha$ as in (\ref{alpha}) and the following expression
for $\chi$:
\begin{equation}
\label{chi}
\chi = 3 \epsilon^{-1/2} \int_{\mathbb{R}} G(x) \hat{\psi}_0^3(x) \hat{\varphi}_0(x) dx.
\end{equation}
The justification of the quintic DNLS equation (\ref{DNLSmodified}) relies on the two facts.

\begin{lemma}
Under condition (\ref{potentials}), there exists a solution $\hat{\varphi}_0(x)$
of the inhomogeneous equation (\ref{inhomogeneous-problem}) so that
$\chi$ is bounded and nonzero as $\epsilon \to 0$.
\label{lemma-1}
\end{lemma}

\begin{proof}
First we note that if $G(x)$ is odd and $\hat{\psi}_0(x)$ is even on $\mathbb{R}$, then
$\hat{\varphi}_0(x)$ is odd on $\mathbb{R}$, so that the integral in (\ref{chi})
is generally non-zero. Moreover, using
the inhomogeneous equation (\ref{inhomogeneous-problem}), we infer that
\begin{equation}
\label{chi-negative}
\chi = -3 \langle (L - \hat{E}_0) \hat{\varphi}_0,\hat{\varphi}_0 \rangle,
\end{equation}
so that $\chi < 0$ for the lowest band of $L$, since $(L - \hat{E}_0)$ is positive
definite if $\hat{E}_0$ is at the bottom of the spectrum of $L$ in the limit $\epsilon \to 0$.
To show that $\chi$ is bounded as $\epsilon \to 0$, we can use again the Gaussian
approximation (\ref{leading-order}) and find solutions of the inhomogeneous equation
(\ref{inhomogeneous-problem}) near $x = 0$ in the form
\begin{equation}
\label{Gaussian-varphi}
\hat{\varphi}_0(x) \sim - \frac{\epsilon^{1/2} G'(0)}{8 (\pi \epsilon)^{3/4}} x e^{-\frac{3 x^2}{2 \epsilon}},
\quad \mbox{\rm near} \quad x = 0.
\end{equation}
As a result,
\begin{equation}
\label{chi-expression}
\chi \sim - \frac{(G'(0))^2}{16 \sqrt{3} \pi},
\end{equation}
and we see that $\chi$ is bounded and negative as $\epsilon \to 0$.
\end{proof}

\begin{lemma}
Under condition (\ref{potentials}),
the largest overlapping integrals from the cubic term $|\Psi_0|^2 \Psi_0$,
\begin{equation}
\label{cubic-overlapping-integral}
\int_{\mathbb{R}} G(x) \hat{\psi}_0^3(x) \hat{\psi}_0(x-2\pi) dx, \quad
\int_{\mathbb{R}} G(x) \hat{\psi}_0^2(x) \hat{\psi}_0^2(x-2\pi) dx,
\end{equation}
are smaller than $\epsilon^{1/2} \mu$.
\label{lemma-2}
\end{lemma}

\begin{proof}
First, we note that if $G(x)$ is a smooth, $2 \pi$-periodic,
and odd function, then $G(x)$ is also odd with respect
to the point $x = \pi$, so that
\begin{equation}
\label{nonlinearity-1}
\int_{\mathbb{R}} G(x) \hat{\psi}_0^2(x) \hat{\psi}_0^2(x-2\pi) dx = 0.
\end{equation}
To consider the other nonzero integral in (\ref{cubic-overlapping-integral}), we write
\begin{equation}
\label{nonlinearity-2}
\int_{\mathbb{R}} G(x) \hat{\psi}_0^3(x) \hat{\psi}_0(x-2\pi) dx =
\left( \int_{-\infty}^{\pi} + \int_{\pi}^{\infty} \right) G(x) \hat{\psi}_0^3(x) \hat{\psi}_0(x-2\pi) dx.
\end{equation}
The second integral on $[\pi,\infty)$ is much smaller than the first integral on $(-\infty,\pi]$
thanks to the faster decay of $\hat{\psi}^3_0(x)$ compared to $\hat{\psi}_0(x)$ on $\mathbb{R}$.
As a result, we deal only with the first integral, which we rewrite as follows:
\begin{eqnarray}
\label{technical-integral}
& \phantom{t} & \int_{-\infty}^{\pi} G(x) \hat{\psi}_0^3(x) \hat{\psi}_0(x-2\pi) dx \\ \nonumber
& = & - \epsilon^{1/2} \int_{-\infty}^{\pi} \hat{\psi}_0(x-2\pi) \left(
-\partial_x^2 + \epsilon^{-2} V_0(x) - \hat{E}_0 \right)
\hat{\varphi}_0(x) dx \\ \nonumber
& = & \epsilon^{1/2} \left[ \hat{\psi}_0(\pi) \hat{\varphi}_0'(\pi) +
\hat{\psi}_0'(\pi) \hat{\varphi}_0(\pi) \right] -
\epsilon^{1/2} \int_{-\infty}^{\pi} \hat{\varphi}_0(x) \left(
-\partial_x^2 + \epsilon^{-2} V_0(x) - \hat{E}_0 \right) \hat{\psi}_0(x-2\pi) dx,
\end{eqnarray}
where we recall again that $\hat{\psi}_0(x)$ is even on $\mathbb{R}$. In view
of equation (\ref{Wannier}), we have
\begin{eqnarray*}
& \phantom{t} & \int_{-\infty}^{\pi} \hat{\varphi}_0(x) \left(
-\partial_x^2 + \epsilon^{-2} V_0(x) - \hat{E}_0 \right) \hat{\psi}_0(x-2\pi) dx \\
& = &
\sum_{n \geq 1} \hat{E}_n \int_{-\infty}^{\pi} \hat{\varphi}_0(x) \left(
\hat{\psi}_{n+1}(x) + \hat{\psi}_{-n+1}(x) \right) dx \\
& \sim & \hat{E}_1 \int_{-\infty}^{\pi} \hat{\varphi}_0(x) \hat{\psi}_0(x) dx =
-\hat{E}_1 \int_{\pi}^{\infty} \hat{\varphi}_0(x) \hat{\psi}_0(x) dx,
\end{eqnarray*}
where the last equality is due to the fact that $\hat{\psi}_0(x)$ is even and $\hat{\varphi}_0(x)$ is odd on $\mathbb{R}$.

Thanks to the fast decay of $\hat{\psi}_0(x)$ and $\hat{\varphi}_0(x)$ on $\mathbb{R}$, 
the second term in (\ref{technical-integral}) becomes smaller than
$\epsilon^{1/2} \hat{E}_1 = \epsilon^{1/2} \mu \alpha$, where $\alpha$ is given by (\ref{alpha}).

Boundary values of $\hat{\psi}_0(x)$,
$\hat{\varphi}_0(x)$ and their derivatives at $x = \pi$ in the first term in 
(\ref{technical-integral})
can again be computed from the WKB solutions for $\hat{\psi}_0(x)$ and
$\hat{\varphi}_0(x)$.

For solutions of the inhomogeneous
equation (\ref{inhomogeneous-problem}), we substitute
$$
\hat{\varphi}_0(x) \sim B(x) e^{-\frac{3}{\epsilon} \int_0^x S(x') dx'}, \quad x \in (0,2\pi),
$$
where $S(x) = \sqrt{V_0(x)}$ and $B(x)$ satisfies the first-order differential equation
$$
\epsilon (6 S(x) B'(x) + 3 S'(x) B(x) - B(x)) - 8 S^2(x)B(x) = - \epsilon^{3/2} G(x) A^3(x),
$$
where the term $B''(x)$ is neglected and $\hat{E}_0 \sim 1/\epsilon$ is used.
Solving the differential equation with the integration
factor, we obtain
$$
B(x) = \frac{C(x)}{S^{1/3}(x)} \exp\left(\frac{4}{3\epsilon} \int_0^x S(x') dx' + \frac{1}{6}
\int_0^x \frac{1 - S'(x')}{S(x')} dx' \right), \quad x \in (0,2\pi),
$$
with
$$
C(x) = -\frac{\epsilon^{1/2}}{6} \int_0^x \frac{G(x') A^3(x')}{S^{1/3}(x')}
 \exp\left(- \frac{4}{3\epsilon} \int_0^{x'} S(x'') dx'' - \frac{1}{6}
\int_0^{x'} \frac{1 - S'(x'')}{S(x'')} dx'' \right) dx'.
$$
Using the Laplace method for computing integrals, we obtain a correct behavior
of $\hat{\varphi}_0(x)$ near $x = 0$ that matches the previous calculation (\ref{Gaussian-varphi}):
\begin{eqnarray*}
\hat{\varphi}_0(x) & \sim & - \frac{\epsilon^{1/2} G'(0)}{6 (\pi \epsilon)^{3/4} x^{1/3}}
e^{-\frac{3 x^2}{2 \epsilon}} \int_0^x y^{1/3} e^{-\frac{2}{3 \epsilon}(y^2 - x^2)} dy \\
& \sim & - \frac{\epsilon^{1/2} G'(0)}{8 (\pi \epsilon)^{3/4}} x e^{-\frac{3 x^2}{2 \epsilon}},
\quad \mbox{\rm near} \quad x = 0.
\end{eqnarray*}
As a result, we have
\begin{eqnarray*}
& \phantom{t} &  \epsilon^{1/2} \left[ \hat{\psi}_0(\pi) \hat{\varphi}_0'(\pi) +
\hat{\psi}_0'(\pi) \hat{\varphi}_0(\pi) \right] \\
& \sim &  \frac{4}{9} S^{2/3}(\pi) A(\pi) \exp\left(-\frac{8}{3\epsilon} \int_0^{\pi} S(x) dx
+ \frac{1}{6} \int_0^{\pi} \frac{1 - S'(x)}{S(x)} dx \right) C_0,
\end{eqnarray*}
where
$$
C_0 = \int_0^{\pi} \frac{G(x) A^3(x)}{S^{1/3}(x)}
 \exp\left(- \frac{4}{3\epsilon} \int_0^{x} S(x') dx' - \frac{1}{6}
\int_0^{x} \frac{1 - S'(x')}{S(x')} dx' \right) dx.
$$
For instance, if $V(x) = 2(1 - \cos(x))$, we obtain
\begin{eqnarray*}
& \phantom{t} & \epsilon^{1/2} \left[ \hat{\psi}_0(\pi) \hat{\varphi}_0'(\pi) +
\hat{\psi}_0'(\pi) \hat{\varphi}_0(\pi) \right] \sim
\frac{4}{9 \pi \epsilon} e^{-\frac{48}{3 \epsilon}} \int_0^{\pi} \frac{G(x)}{\sin^{2/3}(\frac{x}{4})
\cos^{10/3}(\frac{x}{4})} e^{-\frac{16}{3 \epsilon} \cos(\frac{x}{2})} dx,
\end{eqnarray*}
which is clearly smaller than
$$
\epsilon^{1/2} \mu = \frac{1}{\pi^{1/4} \epsilon} \exp\left(-\frac{2}{\epsilon} \int_0^{\pi} S(x) dx \right) =
\frac{1}{\pi^{1/4} \epsilon} e^{-\frac{8}{\epsilon}}.
$$
This completes the proof of Lemma \ref{lemma-2}. 
\end{proof}

Using the approach from Pelinovsky \& Schneider in \cite{Sch2}, we can justify 
the quintic DNLS equation (\ref{DNLSmodified}) on a finite time interval, 
according to the following statement.

\begin{theorem}
Assume that $V(x)$ and $G(x)$ are given by (\ref{potentials}) and 
(\ref{potential-eps}), $G'(0) \neq 0$, and $\mu$ is given by (\ref{mu}).
Let $\{ c_n(T) \}_{n \in \mathbb{Z}} \in C^1(\mathbb{R},l^1(\mathbb{Z}))$
be a global solution of the quintic DNLS equation (\ref{DNLSmodified})
with initial data $\{ c_n(0) \}_{n \in \mathbb{Z}} \in l^2_p(\mathbb{Z})$
for any $p > \frac{1}{2}$. Let $\Psi_0 \in {\cal H}^1(\mathbb{R})$ satisfy the bound
$$
\left\| \Psi_0 - \epsilon^{-1/4} \mu^{1/4}
\sum_{n \in \mathbb{Z}} c_n(0) \hat{\psi}_n \right\|_{{\cal
H}^1(\mathbb{R})} \leq C_0 \epsilon^{-1/4} \mu^{3/4},
$$
for some $C_0 > 0$. There exists $\mu_0 > 0$, $T_0 > 0$, and $C > 0$,
such that for any $\mu \in (0,\mu_0)$, the GP equation (\ref{INLSEdt})
with initial data $\Psi(0) = \Psi_0$ has a solution
$\Psi(t) \in C([0,T_0/\mu],{\cal H}^1(\mathbb{R}))$ satisfying
the bound
$$ 
\forall t \in \left[ 0, T_0/\mu \right] : \quad
\left\| \Psi(\cdot,t) - \epsilon^{-1/4} \mu^{1/4} \left( \sum_{n \in \mathbb{Z}}
c_n(\mu t) \hat{\psi}_n \right) e^{-i \hat{E}_0 t} \right\|_{{\cal H}^1(\mathbb{R})} \leq C
\epsilon^{-1/4} \mu^{3/4}.
$$
\label{theorem-2}
\end{theorem}

\begin{remark}
Note that the results of Theorems \ref{theorem-1} and \ref{theorem-2} 
also hold for the piecewise-constant Kronig-Pennig
potential $V_0(x)$ after minor modifications required because of a different
algebraic factor of $\epsilon$ in the definition of $\mu$ \cite{Sch1}. 
\end{remark}

\section{Localized solutions of reduced equations}

We shall consider the stationary solutions of the quintic DNLS
equation (\ref{DNLSmodified}), where the coefficients $\alpha$ and $\chi$ are computed
asymptotically for the lowest energy band
of $L = -\partial_x^2 + \epsilon^{-2} V_0(x)$ in the
semi-classical limit $\epsilon \to 0$.

Let $c_n(T) = \phi_n e^{-i \Omega T}$ for a real parameter $\Omega$ and
a real-valued sequence $\{
\phi_n \}_{n \in \mathbb{Z}}$ and obtain the stationary quintic
DNLS equation
\begin{equation}
\label{statDNLS}
\alpha (\phi_{n+1} + \phi_{n-1}) + \chi \phi_n^5 = \Omega \phi_n, \quad n \in \mathbb{Z}.
\end{equation}
The hierarchy of overlapping integrals (\ref{overlaping}) implies
that the energy band function $E(k)$ is given at the leading order by
$$
E(k) \sim \hat{E}_0 + 2 \hat{E}_1 \cos(2 \pi k) + ...
$$
Since $k = 0$ is the minimal point of $E(k)$ for the lowest energy band,
we have $\hat{E}_1 < 0$ so that $\alpha < 0$. See also (\ref{alpha}),
where $\alpha < 0$ is computed in the limit $\epsilon \to 0$. On the other hand,
representation (\ref{chi-negative}) implies that $\chi < 0$ for
the lowest band. See also (\ref{chi-expression}), where $\chi < 0$ is computed for
$G'(0) \neq 0$ as $\epsilon \to 0$. The semi-infinite gap corresponds to
the interval $\Omega < 2
\alpha$. 

Localized solutions of the stationary quintic DNLS equation (\ref{statDNLS}) 
can be obtained from a
minimization of the energy functional
$$
H = \sum_{n \in \mathbb{Z}} \left(
\alpha \phi_n \phi_{n+1} + \frac{\chi}{6} \phi_n^6 \right),
$$
subject to a fixed $N = \sum_{n \in \mathbb{Z}} \phi_n^2$.
According to Theorem 2.1 of Weinstein
\cite{Weinstein}, there exists a minimizer of $H$ (called a
ground state) for ${\rm sign}(\alpha) = {\rm sign}(\chi)$ with
${\rm sign}(\Omega - 2 \alpha) = {\rm sign}(\alpha)$. Monotonic
exponential decay of the sequence $\{ \phi_n \}_{n \in
\mathbb{Z}}$ to zero as $n \to \pm \infty$ was shown in Theorem
1.1 of Pankov \cite{Pankov-discrete} (where the cubic DNLS
equation was considered without loss of generality). Note that the
localized solution also exists if ${\rm sign}(\alpha) = -{\rm
sign}(\chi)$ for ${\rm sign}(\Omega + 2\alpha) = -{\rm
sign}(\alpha)$ thanks to the staggering transformation
$$
\phi_n \to (-1)^n \phi_n, \quad \chi \to -\chi, \quad \Omega \to
-\Omega, \quad \alpha \to \alpha,
$$
that leaves solutions of (\ref{statDNLS}) invariant.
Therefore, the localized solution is not monotonically
decaying if ${\rm sign}(\alpha) = -{\rm sign}(\chi)$. For the
semi-infinite gap, we have shown above that $\alpha$ and $\chi$
have equal {\em negative} sign, so that a localized solution
of the stationary quintic DNLS equation (\ref{statDNLS}) 
exists in the semi-infinite gap for $\Omega < 2 \alpha$. 

Consider the stationary GP equation
\begin{equation}\label{statGP}
-\Phi''(x) + V(x) \Phi(x) + G(x) \Phi^3(x) = \omega \Phi(x), \quad x
\in \mathbb{R},
\end{equation}
which is derived from the GP equation (\ref{INLSEdt}) from $\Psi(x,t) = \Phi(x) e^{-i \omega t}$.
Persistence analysis of gap solitons in 
Pelinovsky {\em et al.} \cite{Sch1} gives the following result. 

\begin{theorem}
\label{theorem-3} Let $V(x)$ and $G(x)$ satisfy (\ref{potentials}) and 
(\ref{potential-eps}) and $G'(0) \neq 0$, and $\mu$ is given by (\ref{mu}).
Let $\{ \phi_n \}_{n \in \mathbb{Z}} \in l^1(\mathbb{Z})$
be a ground state of the stationary quintic DNLS equation (\ref{statDNLS}) 
for $\Omega < 2 \alpha$. 
There exists $\mu_0 > 0$ and $C > 0$,
such that for any $\mu \in (0,\mu_0)$, the stationary GP equation (\ref{statGP})
with $\omega = \hat{E}_0 + \mu \Omega$ has a solution
$\Phi \in {\cal H}^1(\mathbb{R})$ satisfying
the bound
$$
\left\| \Phi - \epsilon^{-1/4} \mu^{1/4} \left( \sum_{n \in \mathbb{Z}}
\phi_n \hat{\psi}_n \right) \right\|_{{\cal H}^1(\mathbb{R})} \leq C
\epsilon^{-1/4} \mu^{3/4}.
$$
Moreover, $\phi(x)$ decays to zero exponentially fast as $|x| \to \infty$. 
\end{theorem}

\begin{remark}
One can also prove existence of gap solitons in the semi-infinite gap 
of the GP equation (\ref{INLSEdt}) with potentials (\ref{potentials}) 
in the opposite limit of large-amplitude gap solitons using the Lyapunov--Schmidt reduction method. See 
Sivan {\em et al.} \cite{Sivan} for an example of this technique for the GP equation 
with a periodic linear potential $V(x)$ and a constant nonlinearity coefficient.
\end{remark}

To summarize, from Theorem \ref{theorem-3}, we predict
existence of gap solitons in the semi-infinite gap for any even $V(x)$ and 
odd $G(x)$ with $G'(0) \neq 0$. To illustrate the existence
numerically, we solve the GP equation by using the so-called
imaginary time method
\cite{Perez-Garcia}. As approximations of localized solutions evolve along the imaginary time,
iterations converge to the ground state of the stationary GP equation (\ref{statGP}).

We have developed a Fourier pseudospectral scheme for the
discretization of the spatial derivatives combined with a
split-step scheme for iterations in the imaginary time, see implementation
of this method by Montesinos \& P\'erez-Garc\'ia \cite{Victor}. In other words, solutions of
$$
\partial_{t} U(x,t) = (A+B) U,
$$
with
$$
A = -\partial_{xx}, \quad B= V(x)+G(x)|U|^2,
$$
are approximated from exact solutions of the problems
$\partial_{t} U= A U$ and $\partial_{t} U = B U$. By using the
symmetric (second-order) split-step method,
whose equation is
\begin{equation}\label{solnumerica}
U(x,t+\tau) = e^{\tau A/2}e^{\tau B}e^{\tau A/2}
U(x,t)+\mathcal{O}(\tau^{3}),
\end{equation}
we calculate a localized solution of the stationary GP equation (\ref{statGP})
as $t \to \infty$. Figure
\ref{figura} shows the branch of gap solitons bifurcating to the
semi-infinite gap (left) and a particular profile of the localized
solution (right) that corresponds to the point on the solution branch on the
left.

We note that the numerical scheme we have used has many
advantages. First, it is more accurate than finite-difference numerical methods.
Second, the Fourier transform can be
computed by using the fast Fourier transform. Finally, the
$L^{2}$-norm of the localized solutions is preserved during the time
iterations so that the $L^2$-norm of a gap soliton along the solution branch can be
fixed by the starting approximation.

 \begin{figure}
\epsfig{file=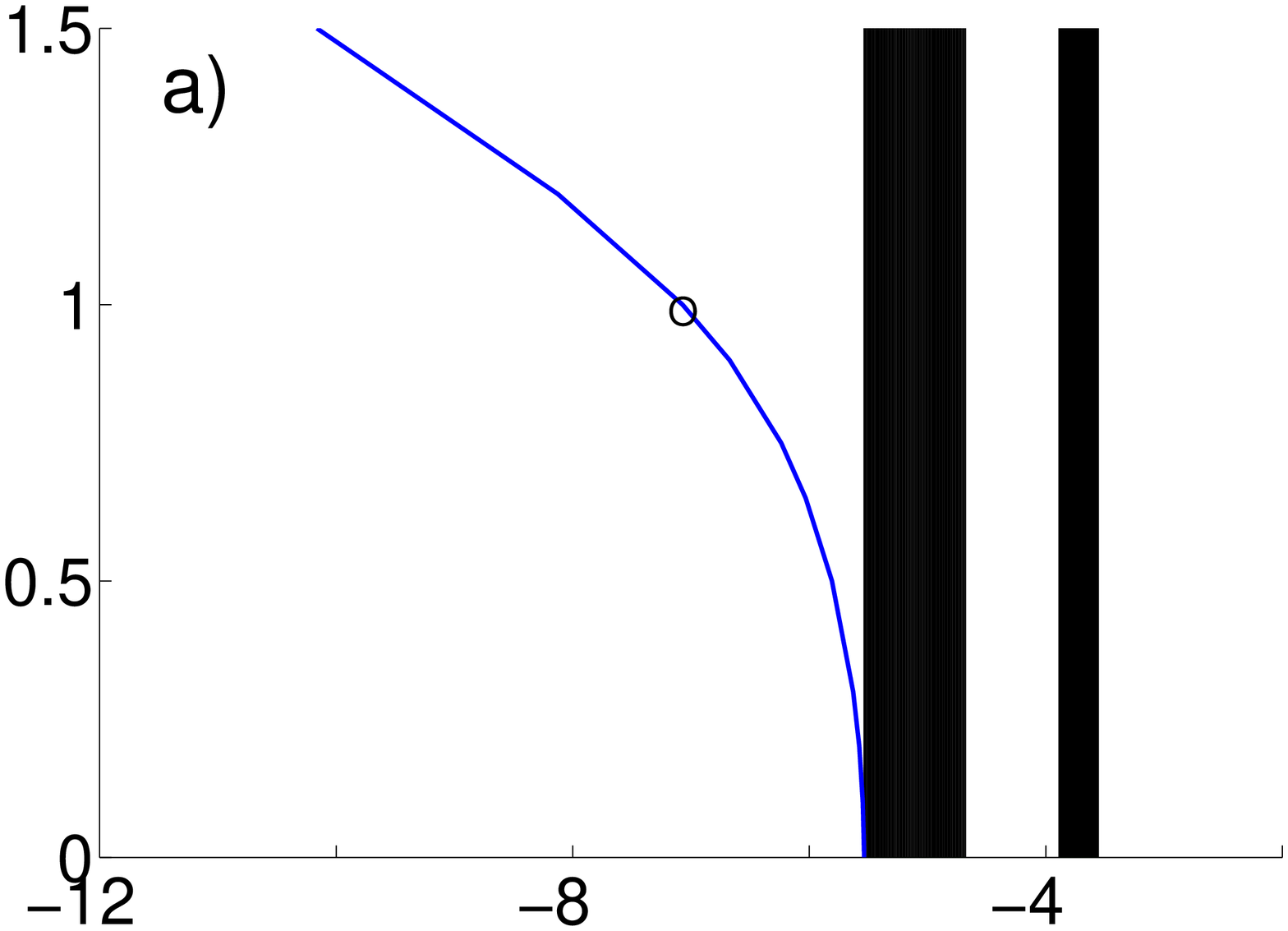,width=7cm}
\epsfig{file=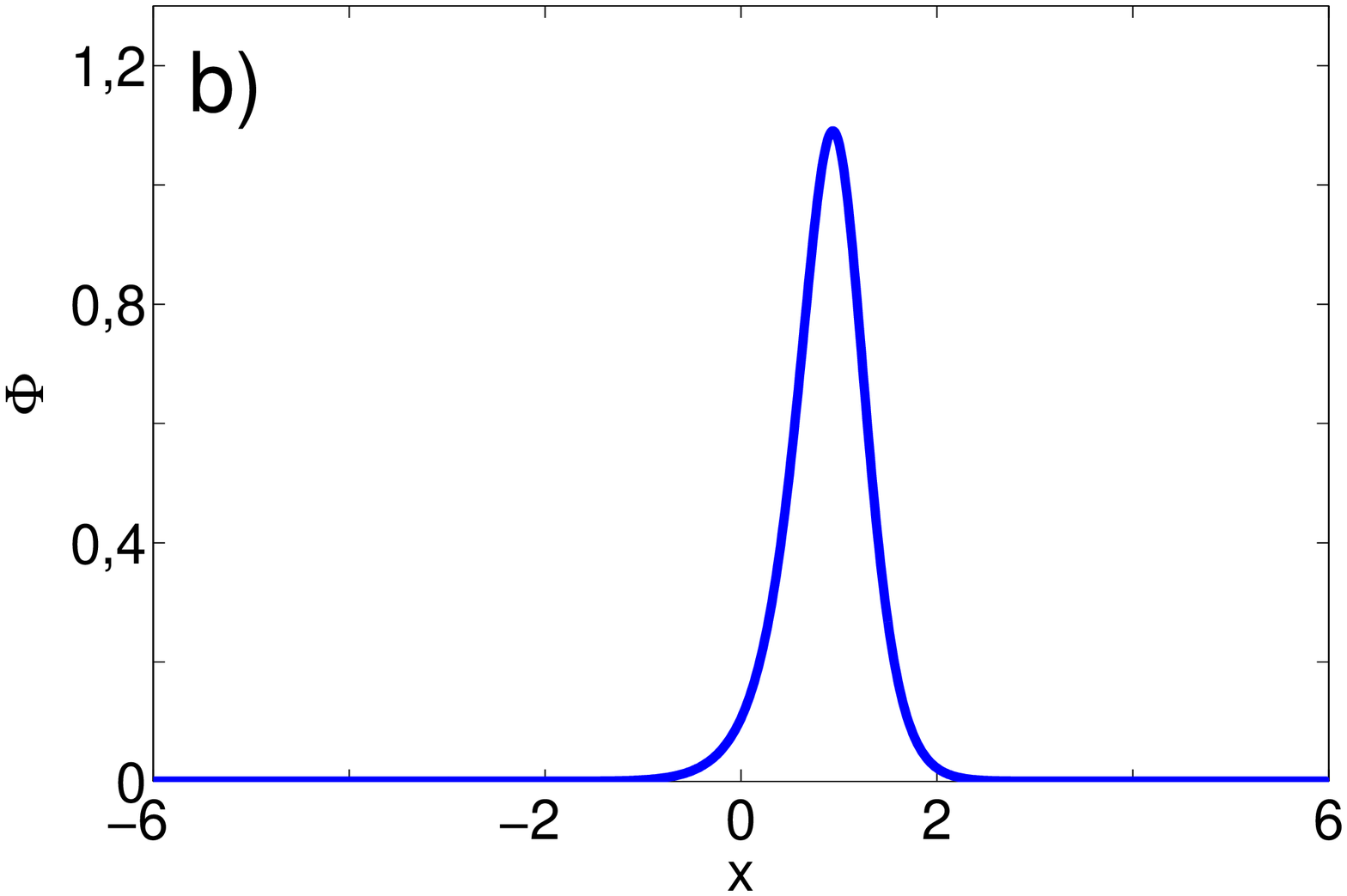,width=7cm} \caption{The solution
family of gap solitons for $G_{0}=-10$ and $V_{0}=6$ in
(\ref{explicit-potential}): The $L^{2}$-norm $N$ versus $\omega$
(left) and the spatial profile of gap soliton corresponding to
marked point with a black circle (right).} \label{figura}
\end{figure}

In the end, we note that existence of localized solutions in the
stationary GP equation (\ref{statGP}) for any
smooth $2\pi$-periodic even $V(x)$ and odd $G(x)$ in the
semi-infinite gap of $L$ can be proved using the variational theory by
a modification of arguments in \cite{Pankov}. This modification
is a subject of an ongoing work \cite{Pankov-personal}. Numerical evidences of 
existence of gap solitons in the semi-infinite gap for sign-varying nonlinearity 
coefficients can be found in \cite{Bludov1,Bludov2}.

\subsection{Comparison with the generalized DNLS equations}

Let us compare our main conclusion with the prediction of the
stationary generalized DNLS equation considered by Abdullaev
{\em et al.} \cite{ABDKK}. For the case of odd
nonlinearity coefficient $G(x)$, this stationary equation is
written in the form
\begin{eqnarray}
\label{statDNLSgen} \alpha (\phi_{n+1} + \phi_{n-1}) + \gamma (3
\phi_n^2 (\phi_{n+1} - \phi_{n-1}) - \phi_{n+1}^3 + \phi^3_{n-1})
= \Omega \phi_n, \;\; n \in \mathbb{Z},
\end{eqnarray}
where $\alpha$ is the same as in (\ref{statDNLS}) and $\gamma$ is
proportional to the overlapping integral (\ref{nonlinearity-2}).
Note that the cubic term in (\ref{statDNLSgen}) is slightly
different from the one in (\ref{DNLS}), which holds for even
nonlinearity coefficient $G(x)$ \cite{ABDKK}. We also note that
$\beta = \delta = 0$ thanks to (\ref{nonlinearity-0}) and
(\ref{nonlinearity-1}). The energy functional is now written as
follows:
$$
H = \sum_{n \in \mathbb{Z}} \left( \alpha \phi_n \phi_{n+1} +
\gamma \phi_n^3 (\phi_{n+1} - \phi_{n-1}) \right).
$$
While we are not able to prove that the stationary DNLS equation
(\ref{statDNLSgen}) admits no localized solutions for any signs of
$\alpha$ and $\gamma$, we can simplify the problem in the slowly
varying approximation, which is also referred to as the continuum
limit of the lattice equation. To this end, we assume that the
following expansion makes sense
$$
\phi_{n\pm 1} = \phi(x_n) \pm h \phi'(x_n) + \frac{1}{2} h^2
\phi''(x_n) + {\cal O}(h^3),
$$
where $x_n = hn$, $n \in \mathbb{Z}$ and apply the scaling
$$
\alpha = 2 h \hat{\alpha}, \quad \Omega - 2 \alpha = 2 h^3
\hat{\Omega}.
$$
At the leading order, the difference equation (\ref{statDNLSgen})
becomes the second-order differential equation
\begin{equation}
\label{second-order-ODE} \hat{\alpha} \phi''(x) - \gamma \phi'(x)
\left[ (\phi'(x))^2 + 3 \phi(x) \phi''(x) \right] = \hat{\Omega}
\phi(x), \quad x \in \mathbb{R},
\end{equation}
which has the first integral
$$
I = \frac{1}{2} \hat{\alpha} (\phi'(x))^2 - \gamma \phi(x)
(\phi'(x))^3 - \frac{1}{2} \hat{\Omega} \phi^2(x).
$$
We note that $I = 0$ for localized solutions and that no turning
point $x_0 \in \mathbb{R}$ with $\phi(x_0) > 0$ and $\phi'(x_0) =
0$ exists. As a result, the trajectory departing from the critical
point $(\phi,\phi') = (0,0)$ in the first quadrant of $(\phi,\phi')$
remains in the first quadrant and goes to infinity. As a result,
no classical localized solutions of the differential equation (\ref{second-order-ODE})
exist. Thus, we have the following result.

\begin{proposition}
Stationary generalized DNLS equation (\ref{statDNLSgen}) for any coefficients 
$\alpha$, $\gamma$, and $\Omega$ admits no localized solutions in the slowly varying 
approximation. 
\end{proposition}

We conclude that the stationary generalized DNLS equation (\ref{DNLS}) gives
the {\em opposite} ({\em wrong}) conclusion to the bifurcation problem of
localized solutions in the semi-infinite gap, compared to the 
stationary quintic DNLS equation  (\ref{statDNLS}).

It is even more problematic how to interpret the modification of the generalized
DNLS equation (\ref{DNLS}) by Smerzi \& Trombettoni \cite{ST}, where the onsite cubic
nonlinear term $\beta |c_n|^2 c_n$ was replaced by $\beta |c_n|^{2
p} c_n$ with $p \leq 2$. If $\beta \neq 0$, the justification of
the cubic DNLS equation (\ref{DNLScub}) in Theorem \ref{theorem-1} leaves no hope to have $p < 2$
in the generalized DNLS equation and to
account the intersite cubic nonlinear terms at the same order as
the onsite cubic nonlinear terms. Thus, we have to conclude that
the generalized DNLS equations considered in \cite{ABDKK} and \cite{ST} (and
implicitly in \cite{CKKS93} and \cite{OJE03}) are invalid for potential $V(x)$ in 
(\ref{potential-eps}) in the tight-binding approximation as $\epsilon \to 0$.

\section{Reductions to the CNLS equation}

Let us now consider the potential $V(x)$ in the GP equation (\ref{INLSEdt})
without assumption (\ref{potential-eps}). Spectral bands are
generally of a finite size, so that we can simplify the GP
equation (\ref{INLSEdt}) if the bound state has small amplitude
near the band edge. This asymptotic reduction leads to the
continuous nonlinear Schr\"{o}dinger (CNLS) equation justified by
Busch {\em et al.} \cite{Sch3}.

To give main details, let $E_0$ be the lowest band edge of operator
$L = -\partial_x^2 + V(x)$ corresponding to the $2\pi$-periodic
$L^2$-normalized eigenfunction $\Psi_0 \in L^2_{\rm per}(0,2\pi)$.
Since the second solution
of $L \Psi = E_0 \Psi$ is linearly growing, the subspace ${\rm
Ker}(L - E_0 I) \subset L^2_{\rm per}(0,2\pi)$ is one-dimensional.
Looking at the Fredholm alternative condition for the
inhomogeneous equation
\begin{equation}\label{eq_psi1}
-\Psi_1''(x) + V(x) \Psi_1(x) - E_0 \Psi_1(x) = 2 \Psi'_0(x),
\end{equation}
we infer that there exists a unique $2\pi$-periodic function
$\Psi_1 \in L^2_{\rm per}(0,2\pi)$ in the orthogonal complement of
${\rm Ker}(L - E_0 I)$. If $V(x)$ is even on $\mathbb{R}$, then
$\Psi_0(x)$ is even and $\Psi_1(x)$ is odd on $\mathbb{R}$. In
addition, if $G(x)$ is an odd $2\pi$-periodic function, there exists
a unique odd $2\pi$-periodic solution of the inhomogeneous
equation
\begin{equation}\label{eq_psinl2}
-\Psi_2''(x) + V(x) \Psi_2(x) - E_0 \Psi_2(x) = - G(x) \Psi_0^{3}(x),
\end{equation}
that also lies in the orthogonal complement of ${\rm Ker}(L - E_0
I)$. Equipped with these facts, we are looking for an asymptotic
solution of the GP equation (\ref{INLSEdt}) using
the decomposition
\begin{eqnarray*}
\Psi(x,t) & = & \varepsilon^{1/2} \left( A(X,T) \Psi_0(x) +
\varepsilon \left( A_X(X,T) \Psi_1(x) + |A(X,T)|^2 A(X,T)
\Psi_2(x) \right) \right. \\
& \phantom{t} & \left. +\varepsilon^{2} \tilde{\Psi}(x,t) \right)
e^{-i E_0 t},
\end{eqnarray*}
where $\varepsilon$ is a small parameter, $X = \varepsilon x$ and
$T = \varepsilon^2 t$ are slow variables, and $\tilde{\Psi}(x,t)$
satisfies the time evolution equation
\begin{eqnarray*}
i \partial_t \tilde{\Psi} & = & (L - E_0) \tilde{\Psi} - i A_T \Psi_0 -
\varepsilon \left( A_{XT} \Psi_1 + ( |A|^2 A )_T
\Psi_2 \right) \nonumber\\ & \phantom{t} & - A_{XX} \left(
\Psi_0 +2 \Psi_1' \right) - 2 (|A|^2 A)_X \Psi_2' - \varepsilon \left( A_{XXX} \Psi_1 +
(|A|^2 A)_{XX} \Psi_2 \right) \\
& \phantom{t} & + G(x) \varepsilon^{-1} \left( |A \Psi_0 +
\varepsilon \left( A_X \Psi_1 + |A|^2 A \Psi_2\right)
+\varepsilon^{2} \tilde{\Psi}|^2  \right. \\
& \phantom{t} & \left. (A \Psi_0 + \varepsilon \left(
A_X \Psi_1 + |A|^2 A \Psi_2\right) +\varepsilon^{2} \tilde{\Psi}) - |A|^2 A \Psi_0^3 \right)
\end{eqnarray*}
Projecting the right-hand side to $\Psi_0$ and truncating at the
leading-order terms, we obtain the CNLS equation
\begin{equation}\label{eq_A}
i A_T = \alpha A_{XX} + \chi |A|^4 A + \gamma \left( |A|^2 A
\right)_X,
\end{equation}
where
\begin{eqnarray*}
\alpha & = & - 1 - 2 \int_{0}^{2\pi} \Psi_1'(x) \Psi_0(x) dx, \\
\chi &=& 3 \int_{0}^{2\pi} G(x) \Psi_0^3(x) \Psi_2(x) dx,\\
\gamma &=& -2 \int_{0}^{2\pi} \Psi_2'(x) \Psi_0(x) dx + \int_0^{2\pi} G(x) \Psi^3_0(x) \Psi_1(x) dx.
\end{eqnarray*}
Justification of the generalized CNLS equation (\ref{eq_A})
can be developed similarly to the work of Busch {\em et al.}
\cite{Sch3}. While it may seem that the generalized CNLS
equation (\ref{eq_A}) contains both the quintic and the cubic derivative terms,
we obtain that
\begin{eqnarray*}
\gamma & = & \int_{0}^{2\pi} \left( -2 \Psi_2'(x)
\Psi_0(x) + G(x) \Psi^3_0(x) \Psi_1(x) \right) dx \\
& = & -\int_{0}^{2\pi} \left( 2 \Psi_2'(x)
\Psi_0(x) + \Psi_1(x) \left( - \partial_x^2 + V(x) - E_0 \right) \Psi_2(x) \right) dx \\
& = & - \int_{0}^{2\pi} \left(  2 \Psi_2(x) \Psi_0'(x) + \Psi_2(x) \left( - \partial_x^2 +
V(x) - E_0 \right) \Psi_1(x) \right) dx =
0.
\end{eqnarray*}
Therefore, the generalized CNLS equation (\ref{eq_A}) is just the quintic CNLS equation
\begin{equation}\label{eq_B}
i A_T = \alpha A_{XX} + \chi |A|^4 A.
\end{equation}

For stationary solutions with $A(X,T) = a(X) e^{-i \Omega T}$,
where $\Omega$ and $a(X)$ are real-valued, we obtain the stationary
quintic NLS equation in the form
\begin{equation}\label{NLSQ}
\alpha a''(X) + \chi a^5(X) = \Omega a(X), \quad X \in
\mathbb{R}.
\end{equation}
Similarly to the case in the tight-binding approximation, we note
that $\alpha < 0$ and $\chi < 0$ for the semi-infinite gap since
$\alpha = -\frac{1}{2} E''(0) < 0$, where $E(k)$ is the energy band
function for the lowest energy band, and
$$
\chi = 3 \int_{0}^{2\pi} G(x) \Psi_0^3(x) \Psi_2(x) dx = -
3\int_0^{2\pi} \Psi_2(x) \left( - \partial_X^2 + V(x) - E_0
\right) \Psi_2(x) dx < 0.
$$
The stationary quintic NLS equation (\ref{NLSQ}) has a positive
definite soliton for ${\rm sign}(\alpha) = {\rm sign}(\chi)$ with
${\rm sign}(\Omega) = {\rm sign}(\alpha)$, that is for $\Omega < 0$ in
the semi-infinite gap.

\end{document}